\pdfoutput=1

\documentclass[11pt,letterpaper]{article}
\renewcommand{\footnoterule}{%
  \kern -3pt
  \hrule width \textwidth height 0.5pt
  \kern 2pt
}
\usepackage[nodisplayskipstretch]{setspace}
\usepackage{amssymb}
\usepackage{amsmath}
\usepackage{amsthm}
\usepackage{accents}
\usepackage{graphicx}
\usepackage{sectsty}
\usepackage{float}
\usepackage{caption}
\usepackage{hyperref}
\usepackage{arydshln}
\usepackage{newtxtext,newtxmath,fancyhdr}  %%%A way to get times. Packages near the top need loading. Seem to lose plus signs (?)  

\newtheorem{theorem}{Theorem}[section]

\newcommand{\PRLsep}{\noindent\makebox[\linewidth]{\resizebox{0.750\linewidth}{1pt}{$\blacklozenge$}}\bigskip}

\sectionfont{\normalsize}
\subsectionfont{\small}

\begin{document}

\pagestyle{fancy}
\fancyhead{} % clear all header fields
\fancyhead[OR]{\thepage}
\fancyhead[OC]{{\footnotesize{\textsf{SOME SINGULAR SPACETIME ALTERNATIVES}}}}
%%\fancyhead[OC]{{\footnotesize{\textsf{CONTINUOUS GRAVASTAR SOLUTIONS WITH KERR EXTERIORS}}}}
%%\fancyhead[OC]{{\footnotesize{\textsf{KERR EXTERIOR GEOMETRIES WITH DESITTER CORES}}}}
%%\fancyhead[OC]{{\footnotesize{ENERGY CONDITION RESPECTING WARP DRIVES IN EINSTEIN-CARTAN THEORY}}}
\fancyfoot{} % clear all footer fields
\renewcommand\headrulewidth{0.5pt}
\addtolength{\headheight}{2pt} % make space for the rule

\title{{\bf{\LARGE{\textsf{Some singular spacetimes and their possible alternatives}}}}\\{\normalsize\textsf{International conference on gravitation, astrophysics and cosmology 2024}}}

\author {
{\small Andrew DeBenedictis  \footnote{adebened@sfu.ca}} \\
\it{\small Faculty of Science, Simon Fraser University} \\
\it{\small and}\\
\it{\small The Pacific Institute for the Mathematical Sciences}\\
\it{\small Burnaby, British Columbia, V5A 1S6, Canada}\\
}
\date{{\small October 6, 2024}}
\maketitle

\setcounter{footnote}{0}
\begin{abstract}
\noindent This is a brief review article of a seminar given at the International Conference on Gravitation, Astrophysics and Cosmology 2024 (ICGAC-2024), Mathura, India. We begin with a historical survey of some singular solutions in the theory of gravitation, as well as a very brief discussion of how black holes could physically form. Some possible scenarios which could perhaps eliminate these singularities are then reviewed and discussed. Due to the vastness of the field the coverage is not exhaustive, but instead the concentration is on a small subset of topics such as possible quantum gravity effects, non-commutative geometry, and gravastars. A simple singularity theorem is also presented in the appendix. Although parts of the manuscript assume some familiarity with relativistic gravitation or differential geometry, the aim is for the broad picture to be accessible to non-specialists of other physical sciences and mathematics. 
\end{abstract}
\rule{\linewidth}{0.2mm}
\vspace{-1mm}

\section{Introduction}

Despite its ubiquity, gravitation remains one of the most enigmatic phenomena in the physical sciences. There is, for example, its universal nature, not shared by the other fundamental forces. There is also its direct connection to causality and causal structure; Gravitation determines what is past, present, and future. In fact, our best theories of gravity so far (of which general relativity still comes out on top even after over a century of scrutiny) indicate that gravity is not a force in the  traditional sense at all. For much of this manuscript it is assumed that gravitation, at least in the low energy arena, is described by the theory of general relativity, which is governed by the Einstein equations
\begin{equation}
 G_{\mu\nu}:=R_{\mu\nu}-\frac{1}{2}R\,g_{\mu\nu} + \Lambda g_{\mu\nu}=8\pi T_{\mu\nu}\,. \label{eq:einsteq}
\end{equation}
For the non-gravitational specialist, in Equation (\ref{eq:einsteq}), $R_{\mu\nu}$ represents the components of the Ricci curvature tensor, $R:=R^{\alpha}_{\;\alpha}$ is its invariant trace (the Ricci scalar), $g_{\mu\nu}$ is the components of the metric (essentially the gravitational potentials), and $T_{\mu\nu}$ is the stress--energy tensor of the gravitating material. The signature of the metric is $+2$ (one negative metric eigenvalue and three positive values) and Greek indices ($0 \rightarrow 3$) spanning the full spacetime, with the zeroth component assumed to be time-like (measuring time) and the $1,\,2,\,3$ components being space-like (measuring space). Roman indices will denote only the spatial components of the tensors. $\Lambda$ is the cosmological constant, which is generally assumed to be zero in this paper unless otherwise indicated. The Einstein tensor $G_{\mu\nu}$ is identically divergence-free such that $\nabla_{\mu}G^{\mu}_{\;\,\nu}\equiv 0$, which via Equation (\ref{eq:einsteq}) implies the \emph{conservation law}
\begin{equation}
 \nabla_{\mu}T^{\mu}_{\;\;\,\nu}=0\,. \label{eq:conslaw}
\end{equation}

In all expressions in this manuscript, the speed of light is set equal to one, as are the gravitational constant and Planck's constant. We concentrate here on the theory of general relativity, and not all results in this manuscript hold if the theory of gravity is not general relativity.

The theory of general relativity has been highly successful. However, like many classical theories in physics, some of its solutions are singular, and these solutions, singularity aside, are not just academic curiosities but physically viable solutions to the equations of motion. It is desirable therefore to try to eliminate these singularities in some physically acceptable way. There is, of course, a long history to this arena of study, and here we only briefly describe a few of the many possible methods of singularity avoidance. It is hoped that this will give the reader some understanding of the problem and some of the methods which can be employed in attempting to rectify it, along with some of the difficulties encountered in the attempt. This manuscript is based on a seminar given at the International Conference on Gravitation, Astrophysics and Cosmology 2024, where the audience consisted of physical scientists and mathematicians. Some familiarity with the topic of relativistic gravity and differential geometry is assumed, but these details can just be accepted by non-specialists, and it is hoped that the overall picture will be accessible to a wider audience of scientists and mathematicians than just gravitational specialists. 

This paper is laid out in two main sections. In Section \ref{sec:hist}, we briefly discuss the history of some important gravitational solutions which possess some sort of singularity. These include the famous Schwarzschild and Kerr solutions and Friedmann--Lema\^{i}tre--Robertson--Walker metrics. We further briefly describe the physical process which could lead to the formation of a black hole. In the following section (Section \ref{sec:alternatives}), we discuss some possible ways that the singularities can be circumvented. An appendix is included which presents a simple singularity theorem illustrating how certain types of singularities can occur.

%%%%%%%%%%%%%%%%%%%%%%%%%%%%%%%%%%%%%%%%%%
\section{A history of some singular solutions in gravity}\label{sec:hist}
Here, we will present a brief history of some of the most important gravitational field solutions which possess some sort of singularity. The review here will be in chronological order, and therefore we will start with the result of Isaac Newton in his famous \textit{Principia}~\cite{ref:principia}
\begin{equation}
 V(r)=-G\frac{M}{r}\,, \label{eq:Newtgravpot}
\end{equation}
where $V$ represents the gravitational potential outside of a spherically symmetric body of a total mass $M$. The above potential leads to a gravitational force on an external mass $m$ of the form
\begin{equation}
 \mathbf{F}=-G \frac{Mm}{r^{2}} \hat{\mathbf{r}}\,. \label{eq:Newtforce}
\end{equation}

From the expressions above, it is clear that there will be an infinite gravitational force at $r=0$ in the case where $M$ is a point mass.

Let us now move on to general relativity, governed by Equation (\ref{eq:einsteq}). As in the Newtonian analysis above, let us concentrate on situations exhibiting spherical symmetry, and we will also assume that the potentials are time-independent. In this scenario, the metric tensor, which as mentioned above encodes the gravitational potentials, has the following components for spherical coordinates:
\begin{equation}
 ds^{2}=g_{\mu\nu}\textrm{d}x^{\mu}\textrm{d}x^{\nu}=-e^{\gamma(r)}\textrm{d}t^{2} + e^{\alpha(r)} \textrm{d}r^{2} +r^{2} \textrm{d}\theta^{2} + r^{2}\sin^{2}\theta\, \textrm{d}\phi^{2}\,. \label{eq:sphmet}
\end{equation}
with $x^{0}=t$, $x^{1}=r$, $x^{2}=\theta$, and $x^{3}=\phi$.
If we calculate the Ricci tensor and Ricci scalar with the metric components in the above line element and use them in the Einstein equation (Equation (\ref{eq:einsteq})), then the following equations result:
{\allowdisplaybreaks\begin{align}
&r^{-2} \cdot \left[1-e^{-(\alpha +\gamma)}
\cdot \partial _1\left(r e^\gamma \right)\right] + 8\pi \,T^1_{\;\;\,1}(r)
= 0\,, \label{eq:spheinstr} \\[0.45cm]
&e^{-\alpha } \cdot \Big[-(1/2)\, \partial_1
\partial _1\gamma - (1/4) (\partial _1\gamma )^2 \nonumber \\[0.1cm]
\;& \qquad\;\;\;+ (1/2)\, r^{-1} \left(\partial _1\alpha -\partial _1 \gamma
\right) + (1/4) \,\partial _1 \alpha \cdot \partial _1 \gamma \Big] + 8\pi \,T^2_{\;\;\,2}(r) = 0\,,
\label{eq:spheinsttheta} \\[0.45cm]
&r^{-2} \cdot \big[1-\partial _1\left(r
e^{-\alpha }\right)\big] + 8\pi \,T^0_{\;\;\,0}(r) = 0\,. \label{eq:spheinstt}
\end{align}}

\hspace*{-1mm}The $G^{3}_{\;\,3}$ equation in this coordinate system is identical to the $G^{2}_{\;\,2}$ equation and thus is not written.

The conservation of stress--energy law (Equation (\ref{eq:conslaw})) only has one non-trivial component in this scenario:
\begin{align}
 \frac{r}{2} \nabla_{\mu}T^{\mu}_{\;\;\,1} = &  -T^2_{\;\;\,2}(r) + (r/2) \cdot \partial_1
T^1_{\;\;\,1} + \big[1+(r/4) \cdot \partial _1 \gamma \big] T^1_{\;\;\,1}(r)
 \nonumber \\[0.1cm]
& - (r/4) \cdot \partial _1 \gamma \cdot T^0_{\;\;\,0} (r)
= 0\,.  \label{eq:sphconslaw}
\end{align}

We will treat two of the components of the matter distribution, encoded in $T^{\mu}_{\;\;\,\nu}$, as being prescribable, since this is the level at which the system of equations is underdetermined. In this scenario, we can solve for the metric functions $\alpha(r)$ and $\gamma(r)$, and the remaining stress-energy components ($T^{2}_{\;\;\,2}=T^{3}_{\;\;\,3}$) as follows:
\begin{enumerate}
\item Prescribe $T^0_{\;\;\,0}(r)$, which is ``minus the energy density in the comoving frame'' here, and solve Equation (\ref{eq:spheinstt}) for  $e^{-\alpha (r)}$.

\item Now that $e^{-\alpha (r)}$ is known, prescribe $T^1_{\;\;\,1}(r)$, which is the ``radial pressure in the comoving frame'' here, and solve the equation made up of the linear combination of  Equation (\ref{eq:spheinstr}) minus Equation (\ref{eq:spheinstt}) in order to find $e^{\gamma (r)}.$

\item The stress--energy tensor component $T^{2}_{\;\;\,2}$ (in this case, the transverse pressure in the comoving frame) will be \emph{defined} by  the conservation law in Equation (\ref{eq:sphconslaw}).
\end{enumerate}
The above procedure, due to J.~L.~Synge \cite{ref:syngebook} and extended by A.~Das yields the following solution (see \cite{ref:dasdebbook} for details, although it is a fairly straight-forward manipulation)
{\allowdisplaybreaks\begin{align}
e^{-\alpha (r)} =\;& 1+\frac{1}{r} \left[8\pi
\int^{r} T^0_{\;\;\,0}(r')\, r^{\prime 2} \textrm{d}r'\right] \,,
\label{eq:alphasol}  \\[0.4cm]
e^{\gamma (r)} =\;& e^{-\alpha (r)} \cdot \exp \left\{8\pi
\int^{r} \big[T^1_{\;\;\,1}(r') - T^0_{\;\;\,0}(r')\big] \cdot
e^{\alpha (r')} \cdot r' \,\textrm{d}r'\right\} , \label{eq:gammasol}
\\[0.4cm]
T^2_{\;\;\,2}(r) \equiv\;& T^3_{\;\;\,3}(r) := (r/2) \cdot \partial _1
T^1_{\;\;\,1} + \big[ 1+(r/4) \cdot \partial _1\gamma \big] \cdot
T^1_{\;\;\,1}(r) \nonumber \\[0.1cm]
\;& \qquad\quad\;\;\; - (r/4) \cdot \partial _1 \gamma \cdot T^0_{\;\;\,0}
(r)\,. \label{eq:angpresssol}
\end{align}}

Now, let us look for the general relativistic analog of the Newtonian gravitational potential (\ref{eq:Newtgravpot}). We recall that the Newtonian gravitational potential was that due to a point mass $M$ located at $r=0$. Let's therefore use such a distribution in (\ref{eq:alphasol}) - (\ref{eq:angpresssol}). That is, let us prescribe
\begin{equation}
 T^{0}_{\;\,0}=-M\frac{\delta(r')}{4\pi {r'}^{2}}\,, \label{eq:schwstress}
\end{equation}
recalling that $T^{0}_{\;\;0}$ in this coordinate system corresponds to minus the energy density in the comoving frame of the gravitating material.
The factor of $4\pi {r'}^{2}$ comes from the definition of the spherical coordinate delta function when the impulse is located at $r'=0$ which is coincident with a boundary of the integration domain\footnote{It may seem that a factor of $\sqrt{|g_{11}|}$ is missing. This is due to the same factor being missing in the measure of the integral in (\ref{eq:alphasol}). A more physical explanation for these factors comes from the ADM mass integral of the spacetime. See, for example, \cite{ref:grqc0703035v1}.}. If one instead uses another common definition where $\int_{0}^{\epsilon} \delta(r')\, \textrm{d}r'=1/2$ instead of $1$ then the delta function is written as $\delta(r)/(2\pi {r'}^{2})$  (see for example \cite{ref:sphdelt}, \cite{ref:bracewell}). Using (\ref{eq:schwstress}) in (\ref{eq:alphasol}) immediately yields
\begin{equation}
 e^{-\alpha(r)}=1-\frac{2M}{r} \,. \label{eq:schwalph}
\end{equation}
Now, according to the solution scheme outlined above, we still need to prescribe the radial pressure $T^{1}_{\;\,1}$. For this, we will appeal to the junction conditions. Namely, Synge's junction condition in static spherical symmetry dictates that the radial pressure be continuous at all points in spacetime. Since away from the point mass, the radial pressure (and all stress--energy components) are zero, we will use Synge's junction condition of radial pressure continuity to tell us that at the point mass the radial pressure should be zero; that is, $T^{1}_{\;\,1}(0)=0$. Admittedly though, it is somewhat dubious to rely on a junction condition at a singular point. With this choice of vanishing $T^{1}_{\;\,1}$, we find from Equation (\ref{eq:gammasol}) that
\begin{equation}
 e^{\gamma(r)}=\left(1-\frac{2M}{r}\right)e^{h_{0}}\,, \label{eq:schwgam1}
\end{equation}
with $h_{0}$ as a constant. This constant can be absorbed in the rescaling of the $t$ coordinate in Equation (\ref{eq:sphmet}), where $e^{h_{0}/2}t \rightarrow t$. At this stage, we have the metric (gravitational potentials) in general relativity for a ``point mass'' at $r=0$:
\begin{equation}
 ds^{2}=-\left(1-\frac{2M}{r}\right)\,\textrm{d}t^{2} + \frac{\textrm{d}r^{2}}{1-\frac{2M}{r}} + r^{2}\,\textrm{d}\theta^{2} + r^{2}\sin^{2}\theta\, \textrm{d}\phi^{2}\,. \label{eq:schwmet}
\end{equation}

The above is the famous Schwarzschild solution. Schwarzschild wrote the line element in $-2$ signature and in a slightly different way in his original paper~\cite{ref:schworig}, and this is illustrated in Figure~\ref{fig:schworig}, which is taken directly from the original publication:
%%\begin{figure}[h!t]
\begin{figure}[H]
\includegraphics[width=1.00\textwidth, clip, keepaspectratio=true]{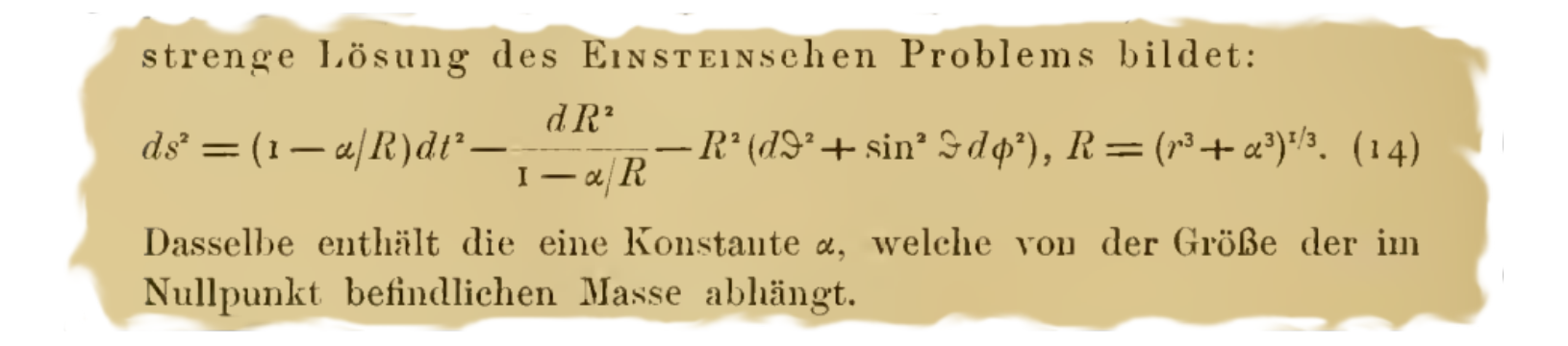}
\caption{An excerpt from K.~Schwarzschild's original paper~\cite{ref:schworig}. {\underline{Please note}} that $R$ in Schwarzschild's paper is analogous to $r$ in this manuscript. Note that negative values of $r$ in this figure are allowed by extension, taking us down to $R=0$ (the singularity).  As we can see in this figure, there is a $C^{2}$ coordinate transformation between Schwarzschild's $R$ coordinate and Schwarzschild's $r$ coordinate ($\partial R/\partial r=r^{2}/(\alpha^{3}+r^{3})^{2/3}$, $\partial^{2}R/(\partial r)^{2}=2\alpha^{3}r/(\alpha^{3} +r^{3})^{5/3}$, $\partial r/\partial R=R^{2}/(R^{3}-\alpha^{3})^{2/3}$, $\partial^{2}r/(\partial R)^{2}=-2\alpha^{3}R/(R^{3}-\alpha^{3})^{5/3}$), save for at the singular point $r=-\alpha$\; ($R=0$), which Schwarzschild did not consider. Figure included with the kind permission of the Berlin-Brandenburg Academy of Sciences and Humanities---Academy Library---Shelfmark: Z 350---1916{,}%MDPI: Please confirm if this should be a dot
1 (Akademiearchiv, Berlin-Brandenburgischen Akademie der Wissenschaften).}
\label{fig:schworig}
\end{figure}

Related to the issue of singularities, we notice that the metric in Equation (\ref{eq:schwmet}) has two singular points (not including the benign issue at the poles $\theta=0,\, \pi$). These are located at $r=2M$ and $r=0$. In order to study these possible pathologies, we can look at the Riemann curvature tensor components, which essentially encode the ``strength'' of the gravitational field via the amount of spacetime curvature it produces. These components are not invariant though and are therefore subject to coordinate effects, meaning that pathologies found in the Riemann tensor components might be due to the coordinates utilized and not a true pathology. In order to alleviate this problem, we shall introduce an orthonormal tetrad, a set of four ($\hat{\alpha}=0,\,1,\,2,\,3$) orthonormal vectors whose components ($\mu=0,\,1,\,2,\,3$) will be denoted as $e^{\mu}_{\;\hat{\alpha}}$. The hatted indices represent orthonormal components, meaning that they are raised and lowered by the orthonormal (i.e., Minkowski) metric. We shall choose for our set of vectors a set pointing in the coordinate directions; that is, we will pick
\begin{align}
 e^{\mu}_{\;\hat{r}} & =\delta^{\mu}_{\;r}{\sqrt{|g^{11}|}}=\delta^{\mu}_{\;r}\sqrt{\left|1-\frac{2M}{r}\right|},\quad  e^{\mu}_{\;\hat{\theta}}={\delta^{\mu}_{\;\theta}}{\sqrt{|g^{22}|}}=\delta^{\mu}_{\;\theta}\frac{1}{r}\,, \nonumber \\
 e^{\mu}_{\;\hat{\phi}} & ={\delta^{\mu}_{\;\phi}}{\sqrt{|g^{33}|}}=\delta^{\mu}_{\;\phi}\frac{1}{r\sin\theta},\quad e^{\mu}_{\;\hat{t}} =\delta^{\mu}_{\;t}{\sqrt{|g^{00}|}}=\delta^{\mu}_{\;t}\frac{1}{\sqrt{\left|1-\frac{2M}{r}\right|}}\,. \label{eq:schwtet}
\end{align}
This set is orthonormal since it satisfies
\begin{equation}
\left[e^{\mu}_{\;\,\hat{\alpha}}e_{\mu\hat{\beta}}\right]= [\eta_{\hat{\alpha}\hat{\beta}}]=\mbox{diag}\left[-1,\,1,\,1,\,1\right]\,, \nonumber
\end{equation}
the Minkowski metric $[\eta_{\hat{\alpha}\hat{\beta}}]$. 

Projecting the Riemann curvature tensor into the orthonormal frame,
\begin{equation}
  R_{\hat{\alpha}\hat{\beta}\hat{\gamma}\hat{\delta}}=R_{\mu\nu\rho\sigma}e^{\mu}_{\;\hat{\alpha}}e^{\nu}_{\;\hat{\beta}}e^{\rho}_{\;\hat{\gamma}}e^{\sigma}_{\;\hat{\delta}}\,, \nonumber
\end{equation}
yields the following singular components (plus those related by index symmetry and components that are finite)
\begin{align}
 R_{\hat{t}\hat{r}\hat{t}\hat{r}} & =-\frac{2M}{r^{3}} = -R_{\hat{\theta}\hat{\phi}\hat{\theta}\hat{\phi}}\nonumber \\
 R_{\hat{t}\hat{\theta}\hat{t}\hat{\theta}} & =\frac{M}{r^{3}} = R_{\hat{t}\hat{\phi}\hat{t}\hat{\phi}}=-R_{\hat{r}\hat{\theta}\hat{r}\hat{\theta}}=-R_{\hat{r}\hat{\phi}\hat{r}\hat{\phi}} \,. \label{eq:schworthriem}
\end{align}
We notice immediately from (\ref{eq:schworthriem}) that $r=0$ possesses a curvature singularity. However, the other potentially problematic surface, $r=2M$, does not. As is well-known, the event horizon ($r=2M$) is non-singular.

Moving now into the 1920's and 30's, scientists wished to apply the then new theory of general relativity to the universe as a whole. A couple of assumptions went into this endeavor: Due to the difficulty of the problem, it was assumed that, on the very largest scales, the universe is homogeneous and isotropic. In the geometric language of general relativity these assumptions mean that the spatial geometry of the universe is of constant curvature. That is, at a fixed instant of time, the universe is constantly curved and the choices for constant curvature are of course zero curvature (flat), constant positive curvature, or constant negative curvature. There are a set of mathematical conditions that describe constant curvature spaces which lead almost directly to the metric compatible with these assumptions. We shall not describe them here but the interested reader may find them in \cite{ref:dasdebbook}. The resulting metric is the one which yields the following line element
\begin{equation}
 {\rm d}s^{2}=-{\rm d}t^{2} + a^{2}(t)\left[\frac{{\rm d}r^{2}}{1-k\,r^{2}} +r^{2\,}{\rm d}\theta^{2} +r^{2}\sin^{2}\theta\, {\rm d}\phi^{2}\right]\,. \label{eq:flrwmet}
\end{equation}
The above famous metric is known as the Friedmann-Lema\^{i}tre-Robertson-Walker metric \cite{ref:F}, \cite{ref:L}, \cite{ref:R}, \cite{ref:W}. The constant $k$ determines whether the $t=$const. spatial slices are flat ($k=0$), negatively curved ($k< 0$), or positively curved ($k>0$). The function $a(t)$ is arguably the most important part of (\ref{eq:flrwmet}), since its behavior tells us how the geometry of the spatial slices evolve with time. 

An orthonormal tetrad may be constructed for this metric and the Riemann curvature tensor's components projected in the orthonormal frame can be computed. We shall not go into the same level of detail as with the Schwarzschild metric, but the important results for us are summarized in the following two components, where the tetrad has been taken along the coordinate directions:
\begin{equation}
 R_{\hat{t}\hat{r}\hat{t}\hat{r}}=- \frac{\ddot{a}(t)}{a(t)}, \quad R_{\hat{r}\hat{\theta}\hat{r}\hat{\theta}}= \frac{\dot{a}^{2}(t) +k}{a^{2}(t)}\,. \label{eq:orthriemflrw}
\end{equation}
We see that if $a(t)$ reaches zero, then we will have a curvature singularity, unless $\ddot{a}(t) \to 0$  and $\dot{a}^{2}(t) \to -k$ as $a(t)\to 0$. This is the famous Big Bang singularity. Of course, if these conditions can be met, or $a(t)=0$ can be avoided, then we can circumvent the singularity. Unfortunately, when solving the Einstein field equations using physically reasonable $T^{\mu}_{\;\;\,\nu}$ components, these conditions generally cannot be avoided, and a Big Bang singularity is predicted at some finite time in the past.

Going back to the Schwarzschild solution now. Recall that we derived the Schwarzschild solution (\ref{eq:schwmet}) under the assumption that the source term was a the stress-energy of a ``point mass'' at $r=0$. The result was a singular spacetime. Nature however may very well not like true point masses. Instead, one should more realistically consider an extended body, like a star or nebula, that collapses under its own gravitational attraction, and see what sort of metric results when Einstein's equations are applied to this scenario. Such a study was first published in 1939 in the seminal paper by J.~R. Oppenheimer and H. Snyder entitled \emph{On Continued Gravitational Contraction} \cite{ref:opsnyd}. 

Due to the extreme complication of the model, Oppenheimer and Snyder assumed that the gravitating material was a pressureless ``dust'', although some analysis in the paper was performed for non-zero pressure. Outside of the spherically collapsing matter, the vacuum region, the metric was taken to be that of Schwarzschild (\ref{eq:schwmet}). This is not really an assumption since it can be fairly easily shown that the Schwarzschild solution is unique in general relativity in describing a spherically symmetric vacuum. This is a result of the Birkhoff-Jebsen theorem \cite{ref:jeb} \cite{ref:birk}. Oppenheimer and Snyder found in their study that to an external observer the matter collapses and asymptotically approaches the radius $r=2M$, with the redshift factor going to infinity as $r=2M$ is approached. Also, photons emitted from the collapsing matter can escape to infinity over a narrower and narrower range of angles as the collapse progresses. Eventually, in finite time in the comoving frame of the matter, no photons from the matter could be emitted which would reach out to an observer beyond $r=2M$. That is, the collapsing star will eventually appear ``dark''. This was the formation of a black hole, although not called that at the time, with an event horizon at $r=2M$. There was no real discussion of the singularity at $r=0$, but its eventual formation, under mild assumptions (in the Oppenheimer-Snyder model the positivity of the energy density) would be mathematically inevitable.

We now fast-forward to 1963. Roy Kerr studied a vacuum gravitational field which could possibly exist due to a rotating mass~\cite{ref:kerr}. It was later discovered by B. Carter~\cite{ref:carter}, R. Penrose~\cite{ref:penrose}, and R. Boyer and R. Lindquist~\cite{ref:boyerlind} that this metric actually describes a spinning black hole, and furthermore, the solution is unique in the realm of pure vacuum gravity. In the Boyer--Lindquist coordinates, the Kerr line element is written as
\begin{align}
ds^{2}& = g_{\mu\nu}dx^{\mu}dx^{\nu} =-\frac{a^{2} \cos^{2}\theta-2 M  r +r^{2}}{a^{2} \cos^{2}\theta+r^{2}}\,\textrm{d}t^{2}  + \frac{a^{2} \cos^{2}\theta+r^{2}}{a^{2}-2 M  r +r^{2}}\, \textrm{d}r^{2} \nonumber \\
&+\left(a^{2} \cos^{2}\theta+r^{2}\right)\, \textrm{d}\theta^{2} -\frac{4 M  a r \sin^{2}\theta}{a^{2} \cos^{2}\theta+r^{2}}\, \textrm{d}t\, \textrm{d}\phi \nonumber \\
&+ \frac{\sin^{2}\theta \left[\left(a^{2}+r \left(r -2 M  \right)\right) a^{2} \cos^{2}\theta+r \left(\left(r +2 M  \right) a^{2}+r^{3}\right)\right]}{a^{2} \cos^{2}\theta+r^{2}}\, \textrm{d}\phi^{2}\,. \label{eq:kerrmet}
\end{align}

Here, $a$ represents the angular momentum per unit mass, a direct measure of the black hole's spin. The above line element reduces to the Schwarzschild one in Equation (\ref{eq:schwmet}) when $a=0$.

Regarding the possible existence of curvature singularities, we could compute the Riemann tensor components and project them onto an orthonormal frame. However, the results are unwieldy. Instead we calculate the Kretschmann scalar of the Kerr metric,\footnote{If the Kretschmann scalar is infinite there is an infinite pathology in the curvature. If it is finite, there still might be a curvature singularity. In this sense the Kretschmann scalar is a less robust measure of curvature singularities than studying the individual components of the Riemann tensor in an orthonormal frame.} 
\begin{equation}
 K:=R_{\mu\nu\rho\sigma}R^{\mu\nu\rho\sigma} = \frac{48 M^{2}\left ( r^{6} -15a^{2}r^{4}\cos^{2}\theta +15a^{4}r^{2}\cos^{4}\theta - a^{6}\cos^{6}\theta\right)}{\left(r^{2}+a^{2}\cos^{2}\theta\right)^{6}}\,. \label{eq:kerrkretsch}
\end{equation}
At $\theta=\pi/2$, we have
\begin{equation}
 K=\frac{48M^{2}}{r^{6}}\,,
\end{equation}
heralding a curvature singularity in the Kerr black hole at $r=0$, $\theta=\pi/2$.

Black hole-like objects were predicted, although not in their relativistic guise, by John  Michell in a work written in 1783 and printed in 1784~\cite{ref:michell} and by Laplace in 1796 in his \textit{Exposition du Syst\`{e}me du Monde}~\cite{ref:laplace}. An excerpt from Michell's work is shown in Figure~\ref{fig:michell}.
\begin{figure}[htbp]
%%\begin{figure}[t!]
\begin{center}
\includegraphics[width=0.85\textwidth, clip, keepaspectratio=true]{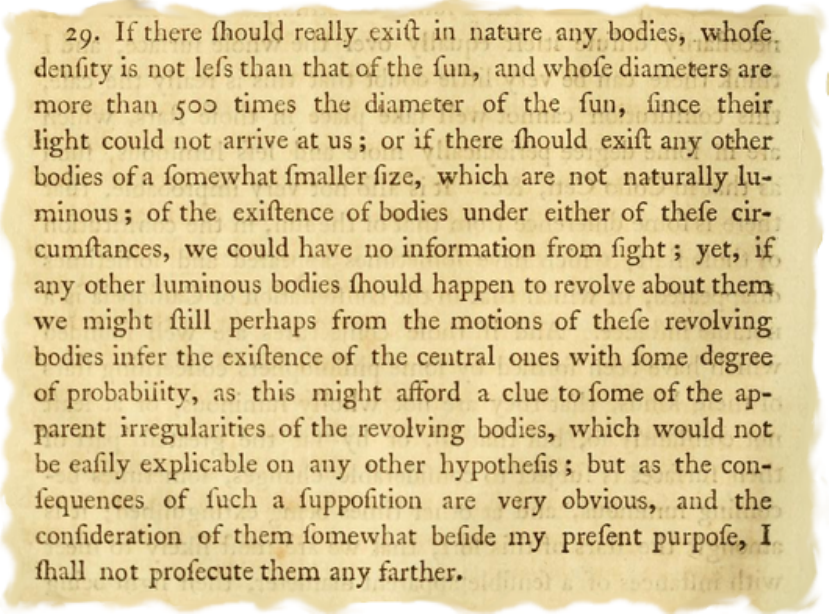}
\caption{{\small{An excerpt from Michell's original 1784 paper \cite{ref:michell}: \emph{If there should really exist in nature any bodies, whose density is not less than that of the sun, and whose diameters are more than 500 times the diameter of the sun, since their light could not arrive at us; or if there should exist any other bodies of a somewhat smaller size, which are not naturally luminous; of the existence of bodies under either of these circumstances, we could have no information from sight; yet, if any other luminous bodies should happen to revolve about them we might still perhaps from the motions of these revolving bodies infer the existence of the central ones with some degree of probability, as this might afford a clue to some of the apparent irregularities of the revolving bodies, which would not be easily explicable on any other hypothesis; but as the consequences of such a supposition are very obvious, and the consideration of them somewhat beside my present purpose, I shall not prosecute them any further.} Copyright Status: Not in copyright.}}}
\label{fig:michell}
\end{center}
\end{figure}

Before ending this historical exposition we mention here that the term black hole for these objects is often attributed to J.~A. Wheeler, around 1967 or 1968. However, this phrase was in  existence at least since 1964 in Science News \cite{ref:scinews}. The phrase was also allegedly used at the 1963 Texas Symposium on Relativistic Astrophysics \cite{ref:symposium} and Wheeler may have been present at the meeting, so it still could have been Wheeler who coined the term black hole. Others attribute the phrase to R. Dicke \cite{ref:whonamedbh}.

\subsection{Astrophysical formation of singular spacetimes}\label{subsec:formation}
The above discussion mainly described the mathematics of certain singular solutions. However, it is important to keep in mind that if the singularities discussed have no way of physically forming, then arguably from a physics perspective they are purely of an academic nature. For this we will look at a very simplified version of the end of a main sequence star's life, keeping in mind though that not all stars will follow this path. 

A star will eventually run out of nuclear fuel and reach a state where its pressure can no longer counteract gravity to maintain dynamic equilibrium. The star will begin to contract under its own weight. Such a star may supernova, and this may leave behind a (non singular) white dwarf, which is a body supported by degenerate electron matter. In some cases though the remnant core may be massive enough that the degeneracy pressure is insufficient to support the core, and further contraction will occur. The mass threshold where this occurs is called the Chandrasekhar mass \cite{ref:chandra} which is approximately (with $G$, $\hbar$ and $c$ not set to unity)
\begin{equation}
M_{Chand}\approx \frac{\alpha_{0}}{2 (\mu_{e}m_{H})^{2}} \sqrt{3\pi} \left(\frac{\hbar c}{G}\right)^{3/2} \label{eq:chandmass}
\end{equation}
Here $\mu_{e}$ is the average molecular weight per electron, $m_{H}$ is the mass of a hydrogen atom, and $\alpha_{0}$ is a constant that arises via the Lane–Emden equation of astrophysics. 

What happens beyond the Chandrasekhar limit is that protons and electrons combine via the reaction
\begin{equation}
 p^{+}+e^{-} \to n +\nu_{e}\,, \label{eq:pecombine}
\end{equation}
forming neutrons and electron-neutrinos. This combination allows further contraction to form a neutron star, now supported by neutron degeneracy. However, if the neutron star is massive enough, exceeding the Oppenheimer-Volkoff limit \cite{ref:oppvol}, even this will not halt the collapse. If this happens other degeneracies could perhaps support the star (e.g.. quark degeneracy), but still, above some mass limit continued collapse will be inevitable. In general relativity the problem is augmented by the fact that the pressure itself, in the stress--energy tensor, is a source of gravity, and not simply just the mass density as in Newtonian gravity. Regardless of the scenario, it seems that above some pressure and energy density limit, black hole formation is inevitable unless some exotic physics comes into play.

%%%%%%%%%%%%%%%%%%%%%%%%%%%%%%%%%%%%%%%%%%
\section{Possible resolutions to singular spacetimes}\label{sec:alternatives}
We will present in this section some possible alternatives to the types of singular spacetimes discussed in the previous section. As mentioned previously, this discussion and coverage are far from exhaustive. 

\subsection{Quantum gravity effects} \label{subsec:qg}
One issue with the field of gravitation is that practical experiments to measure gravitational effects are notoriously difficult to perform. The main reason for this is that gravity itself is so weak. One cannot set up an experiment in a laboratory to directly measure gravity's strong field effects. The singularities we have been discussing here are heralded by an infinity in the (orthonormal) curvature, and at a large curvature, the gravitational field is strong. Now, the theory of general relativity has been thoroughly tested in the weak field limit, such as in the realm of the solar system~\cite{ref:grtests}. It is also becoming fairly well tested in the stronger field regime~\cite{ref:bhgrtests}. However, when we are talking about singularities, we are in the realm of extremely strong gravity. It is possible that in this regime, general relativity is not the correct theory of gravity, and instead, some other theory takes over. An obvious candidate would be the theory of quantum gravitation. Unfortunately, to date, no such fully developed theory exists, although there are a number of candidates. The quantization of general relativity is fraught with problems. In the canonical 3 + 1 formulation, it is plagued by divergences which cannot be addressed in the same way as, say, the theory of quantum electrodynamics. 

One promising candidate theory of quantum gravity is loop quantum gravity (LQG) \cite{ref:lqgrev}. The theory is essentially a theory of quantum geometry and is a direct attempt at quantizing general relativity. At extremely high energies, perhaps near the Planck scale, it is expected that the effects of a quantum theory of gravity manifest themselves. However, not having a fully workable theory at hand, one can look for an effective classical theory that encompasses some of the quantum effects. 

We will be dealing within a symmetry reduced theory. That is, we will assume from the start that the metric has a form similar to (\ref{eq:schwmet}), but the metric functions will not be the same. We write the line element (\ref{eq:schwmet}) in new variables as:
\begin{equation}
 ds^{2}=g_{\mu\nu}dx^{\mu}\textrm{d}x^{\nu} = -N^{2}\, \textrm{d}\tau^{2} +\frac{(E_{2})^{2}}{E_{3}} \textrm{d}y^{2} + E_{3}\,\textrm{d}\varrho^{2} + E_{3} c_{0}\sinh^{2}(\sqrt{d_{0}}\varrho)\,\textrm{d}\varphi^{2}\,, \label{eq:tetradline}
\end{equation}
where $N$, $E_{2}$ and $E_{3}$ are functions of \emph{time}. They are functions of time because, examining (\ref{eq:schwmet}), the classical singularity is located in the domain $r<2M$ inside the event horizon. In the coordinate chart of (\ref{eq:schwmet}) $g_{rr}<0$ and $g_{tt}>0$ for $r<2M$ so that the $r$ coordinate is actually timelike (measures time) and the $t$ coordinate is spacelike (measures distances). The Schwarzschild metric for $r < 2M$ is therefore actually time dependent, and homogeneous. Since we are going to be dealing with domains close to the singularity here, the metric functions are taken as time dependent, and we denote this time as $\tau$.

The $c_{0}$ and $d_{0}$ factors in (\ref{eq:tetradline}) control the geometry of the $\tau=$const. and $y=$const. subspaces. These are the cases: 
\begin{enumerate}
\item $c_{0}=-1$, $d_{0}=-1$: For this choice the $(\varrho,\,\phi)$ sub-manifolds are spheres.
\item $\underset{d_{0}\rightarrow 0}\lim\,c_{0}=\frac{1}{d_{0}}$, $d_{0}=0$: In this case $(\varrho,\,\phi)$ sub-manifolds are toroidal (and the sub-manifolds for this case are intrinsically flat). 
\item $c_{0}=1$, $d_{0}=1$: In this case $(\varrho,\,\phi)$ sub-manifolds are surfaces of constant negative curvature of genus $g > 1$. Such surfaces may be compact or not \cite{ref:topoend}, \cite{ref:nakahara}. 
\end{enumerate}
% 
% \begin{align}
% d\sigma^{2} & = d\theta^{2} + {\sin}^{2} \theta\, d\varphi^{2}  && \mbox{spherical geometry} \nonumber \\
% d\sigma^{2} & = d\varrho^{2} + \varrho^{2}\, d\varphi^{2}  && \mbox{toroidal geometry} \nonumber \\
% d\sigma^{2} & = d \theta^{2} + {\sinh}^{2} \theta\, d\varphi^{2}   && \mbox{hyperbolic geometry}\, .\nonumber 
% \end{align}

We will be dealing within the Hamiltonian formalism and we omit the technical details save for the few that are required to understand the results. For the scenarios above, the gravitational Hamiltonian is given by
\begin{equation}
S=-\frac{N\sqrt{c_{0}}}{2\sqrt{d_{0} E_{3}} \gamma^{2}}\left[\left((a_{2})^{2}-\gamma^{2}d_{0}\right)E_{2}+2E_{3} a_{2}a_{3}\right]
+\frac{N}{2\sqrt{E_{3}}\gamma^2}\Lambda E_{2}E_{3}\,.\label{eq:ourhamconst}
\end{equation}
Here $\Lambda$ is the cosmological constant (required for the toroidal and higher genus scenarios) and the $a_{i}$ which are functions of time, are components of an $su(2)$ connection which are conjugate to the $E_{i}$ in the same way that positions and canonical momenta are conjugate to each other in Hamiltonian mechanics. $\gamma$ is a number, known as the Barbero-Immirzi parameter, which can be set by studying black hole entropy (see, for example, \cite{ref:immirzi}).

We will apply some quantum corrections inspired from LQG to the above Hamiltonian. Specifically we will consider here what are known as holonomy corrections \cite{ref:holonomycors1}, \cite{ref:holonomycors2}. This amounts to making the following substitution in (\ref{eq:ourhamconst}):
\begin{equation}
a_{i} \rightarrow \frac{\sin(a_{i} \delta_{i})}{\delta_{i}}\,.  \label{eq:holonomysubs}
\end{equation}
There are several schemes in the literature for how to pick the quantities $\delta_{i}$. In early studies $\delta_{i}$ was taken to be constant \cite{ref:deltaconst}, \cite{ref:modesto}. Other schemes were later developed that improved the dynamics \cite{ref:deltaconst}, \cite{ref:improvedholonomy}, \cite{ref:chiou}, \cite{ref:reviews}. In this summary we shall present the results in \cite{ref:ourlqgpaper} where the $\delta_{i}$ are given by those similar to the improved dynamics in \cite{ref:deltaconst}:
 \begin{equation}
 \delta_{2}=\left(\frac{\Delta}{E_{3}}\right)^{1/2}\,, \;\;\;\; \delta_{3}=\frac{\sqrt{E_{3}\, \Delta}}{E_{2}}\,, \label{eq:deltapresc}
\end{equation}
where $\Delta$ is taken to be be a constant (usually chosen to be roughly the Planck area).

Substituting (\ref{eq:deltapresc}) into (\ref{eq:holonomysubs}), and then replacing the $a_{i}$ in (\ref{eq:ourhamconst}) with (\ref{eq:holonomysubs}) yields the quantum corrected Hamiltonian. It is then a matter of evolving the variables $E_{i}$ and $a_{i}$ according the Hamiltonian equations of motion $\dot{a}_{2}=\left\{a_{2},\,S\right\}$, etc. In particular, what is of interest is the variable $E_{3}$ since $E_{3}=0$ would herald a curvature singularity analogous to the one in the Schwarzschild metric at $\tau = 0$ (which was $r=0$ in our discussion of the Schwarzschild metric earlier). Qualitatively, what is happening  when $E_{3}=0$ is the following: Note from the line-element (\ref{eq:tetradline}) that volumes of the $(\varrho,\,\phi)$ sub-manifolds are given by
\begin{equation}
V_{\mbox{\footnotesize{$(\varrho,\,\phi)$}}} = E_{3}(\tau) \int \sqrt{c_{0}}\sinh(\sqrt{d_{0}}\varrho)\, \textrm{d}\varrho\,\textrm{d}\phi\,,
\end{equation}
and if the pre-factor $E_{3}(\tau)$ goes to zero the volume elements of this space in a sense ``infinitely compress'' to zero volume resulting in the singular behavior. More mathematically rigorous is that some of the components of the Riemann tensor in an orthonormal frame of the effective metric (\ref{eq:tetradline}) become infinite when $E_{3}(\tau)=0$. 

The time evolution of $E_{3}$ for the case of the loop quantum corrected spherical scenario (with $\Lambda=0.001$) is illustrated in Figure~\ref{fig:E3LQG}. The initial conditions away from the classical singularity are chosen to be the classical ones. That is, the ones that satisfy the general relativity Schwarzschild solution as well as a required constraint condition $S=0$. Also in Figure~\ref{fig:E3LQG} the Hamiltonian evolution of $E_{3}$ for the non-quantum corrected Hamiltonian (i.e. pure general relativity) is illustrated. It can be seen that at $\tau=0$ the uncorrected $E_{3}$ goes to zero, and hence there is a curvature singularity. However, the quantum corrected Hamiltonian evolution avoids $E_{3}=0$, thus avoiding the classical curvature singularity. The full results are discussed in \cite{ref:ourlqgpaper}.
\begin{figure}[htbp]
%%\begin{figure}[t!]
\begin{center}
\includegraphics[width=0.90\textwidth, height=8.0cm,  clip, keepaspectratio=false]{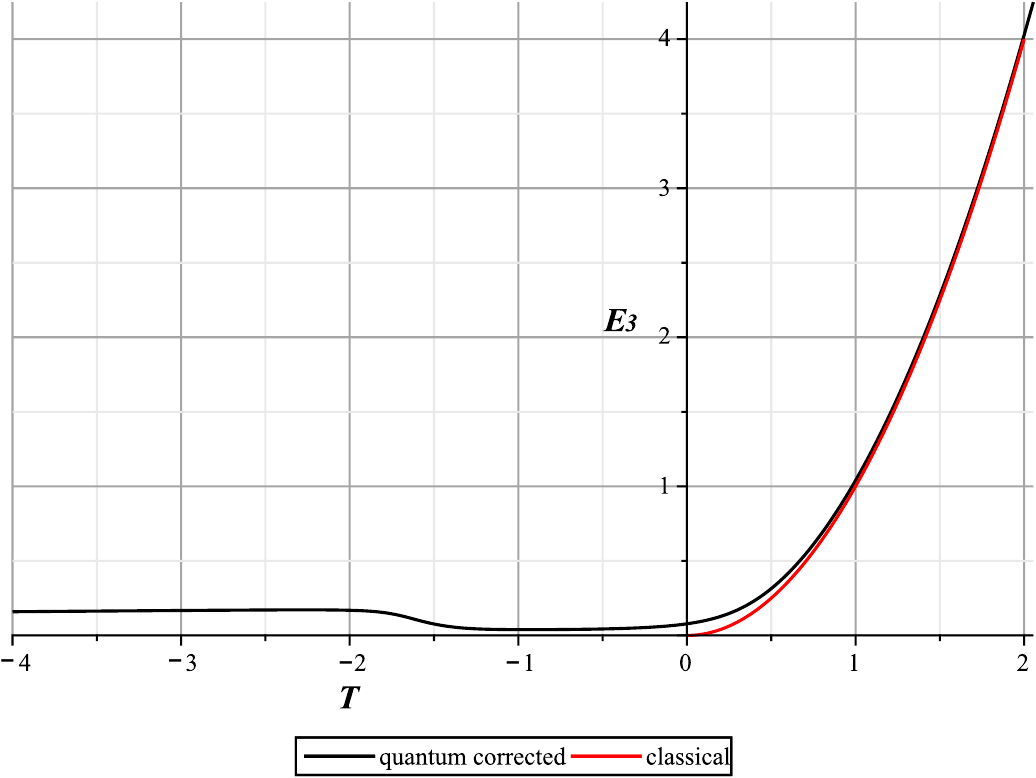}
\caption{{\small{The Hamiltonian evolution of the $E_{3}$ factor in the classical and quantum corrected scenarios. The classical evolution (red), which is the Schwarzschild metric with a small cosmological constant, behaves as $E_{3}=\tau^{2}$, and there is a singularity ($E_{3}=0$) at $\tau=0$. The quantum corrected evolution (black) avoids $E_{3}=0$ and has no corresponding singularity. Parameters: $\gamma=0.274$, $M=5$, $\Lambda= 0.001$. The full results are discussed in \cite{ref:ourlqgpaper}.}}}
\label{fig:E3LQG}
\end{center}
\end{figure}

We close this section by pointing out some drawbacks to this study. First, we should be clear that as of the writing of this manuscript there is no actual fully developed theory of quantum gravity. The study here simply implements some low-order corrections from the theory of loop quantum gravity, and there are other corrections that can be applied \cite{ref:invvol}. Also, the scheme here is ambiguous in the step (\ref{eq:deltapresc}), and other schemes may be employed \cite{ref:improvedholonomy}, \cite{ref:reviews},  \cite{ref:impdyn}. 

\subsection{Non-commutative geometry} \label{subsec:noncomm}
Here we discuss another possible theory which can avoid black hole singularities. It is an implementation of non-commutative geometry \cite{ref:connes},  \cite{ref:noncommgeo}. In quantum mechanics non-commutative geometry refers to extending the usual commutator between positions and momenta, $[x^{j},\, p_{k}]=i\delta^{j}_{k}$ to the positions
 \begin{equation}
[x^{j},\, x^{k}]=i\theta \epsilon^{jk}_{\;\;\,\ell}x^{\ell}\,. \label{eq:noncomcoords}
\end{equation}
Further, this can be extended to non-commutativity between the momenta in an analogous manner
 \begin{equation}
[p_{j},\, p_{k}]=i\beta\epsilon_{jk}^{\;\;\,\ell}p_{\ell}\,. \label{eq:noncommoms}
\end{equation}
In (\ref{eq:noncomcoords}) and (\ref{eq:noncommoms}) the constants $\theta$ and $\beta$ measure the amount of non-commutativity, and in quantum theory would be related to $\hbar$. The full technical details of the discussion in this section maybe found in \cite{ref:schneidnoncomm}.

Going ``backwards'' to classical Hamiltonian mechanics, we replace the commutator of quantum theory with the Poisson bracket of classical mechanics. For particle mechanics this means the analog of (\ref{eq:noncomcoords}) and (\ref{eq:noncommoms}) would be
\begin{subequations}
{\allowdisplaybreaks\begin{align}
 \left\{p_{a},\, x^{b}\right\} = & \, \delta_{a}^{b}\,, \label{eq:poisson1}\\
 \left\{x^{a},\, x^{b}\right\} = & \, \epsilon^{ab}_{\;\;\;c}\theta^{c} \,, \label{eq:poisson2}\\
 \left\{p_{a},\, p_{b}\right\} = & \, \epsilon_{ab}^{\;\;\;c}\beta_{c}\,. \label{eq:poisson3}
\end{align}}
\end{subequations}
Reviews of non-com\-mut\-a\-tive mechanics may be found in \cite{ref:djemai} and \cite{ref:gouba} and references therein. Specifically, the Poisson generalization to non-commutative mechanics is given by the Moyal product
\begin{equation}
 \left\{f,\, g\right\}:= f\star g -g\star f\, ,\label{eq:moyalpoisson} 
\end{equation}
where the Moyal product here is defined as
\begin{equation}
 \left(f\star g\right) (v):=\exp\left[\frac{1}{2}w^{ab} \partial_{a}\tilde{\partial}_{b}\right]f(v)\,g(\tilde{v})_{|\tilde{v}=v}\,, \label{eq:explicitmoyal}
\end{equation}
and operators with a tilde operate only on tilde coordinates and un-tilded operators operate on nontilded coordinates. The two sets of coordinates are made coincident in the end. $w^{ab}$ represents the generalization of the symplectic form whose matrix is given by
\begin{equation}
 \left[w^{ab}\right]=\left[
    \begin{array}{c;{2pt/2pt}c}
        \begin{matrix}
       \epsilon^{ab}_{\;\;\,c}\theta^{c}
        \end{matrix}  &
        \begin{matrix}
        -\delta^{ab} 
        \end{matrix} \\ \hdashline[2pt/2pt]
        \begin{matrix}
       \delta^{ab} 
        \end{matrix} &
         \begin{matrix}
        \epsilon^{ab}_{\;\;\,c}\beta^{c}
        \end{matrix}
    \end{array}
\right]\,. \label{eq:sympform}
\end{equation}

Next we transition to field theories from the particle mechanics picture above. In field theories the field variable is the configuration variable ($\psi^{a}(y)$) and associated with this is a canonical field momentum ($\pi_{a}(x)$). Analogous brackets to (\ref{eq:poisson1}) - (\ref{eq:poisson3}) are implemented:
\begin{subequations}
{\allowdisplaybreaks\begin{align}
 \left\{\pi_{a}(x),\, \psi^{b}(y)\right\} \propto & \, \delta_{a}^{b}\delta(x-y)\,,  \label{eq:fieldpoisson1}\\
 \left\{\psi^{a}(x),\, \psi^{b}(y)\right\} \propto & \, \epsilon^{ab}_{\;\;\;c}\theta^{c}\delta(x-y)  \,, \label{eq:fieldpoisson2}\\
 \left\{\pi_{a}(x),\, \pi_{b}(y)\right\} \propto & \, \epsilon_{ab}^{\;\;\;c}\beta_{c}\delta(x-y)\,. \label{eq:fieldpoisson3}
\end{align}}
\end{subequations}

In gravitation this procedure is not necessarily straight-forward \cite{ref:ncgeinststart}-\cite{ref:ncgeinstend}. To make the problem somewhat tractable the symmetry is frozen so that the metric is fixed to the form (\ref{eq:tetradline}). The system then mimics a mechanical model and it can be argued that the Poisson algebra be augmented in a similar manner to (\ref{eq:fieldpoisson1}) - (\ref{eq:fieldpoisson3}). For the variables in (\ref{eq:tetradline}) these become
\begin{equation}
 \left\{E_{i},\,a_{j}\right\}=\delta_{ij}, \;\;\; \left\{a_{i},\,a_{j}\right\}=\epsilon_{ij}\theta\,, \;\;\; \left\{{E}_{i},\,{E}_{j}\right\}=\epsilon_{ij}\beta\,, \label{eq:minisuppoisson}
\end{equation}
with $\epsilon_{ij}$ the two-dimensional antisymmetric symbol. The Hamiltonian (\ref{eq:ourhamconst}) is then used to evolve the system $\dot{a}_{2}=\left\{a_{2},\,S\right\}$ etc. with the crucial difference from regular Hamiltonian evolution being that that the brackets now include the extended non-commutativity of (\ref{eq:minisuppoisson}). As in section \ref{subsec:qg}, the critical variable for singularity studies is $E_{3}$, with the classical singularity corresponding to $E_{3}=0$ at $\tau=0$. In Figure~\ref{fig:noncommsing} one of many possible evolutions are illustrated for the case of a spherical black hole. The initial conditions are set in the same manner as in section \ref{subsec:qg}. We notice from the figure that the singularity at $\tau=0$, where in the regular theory $E_{3}\to 0$, is delayed indefinitely in the non-commutative mechanics.

\begin{figure}[htbp]
%%\begin{figure}[t!]
\begin{center}
\includegraphics[width=1.00\textwidth, clip, keepaspectratio=true]{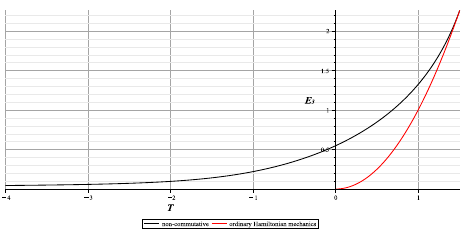}
\caption{{\small{The Hamiltonian evolution of the $E_{3}$ factor in regular Hamiltonian mechanics (red) and the noncommutative theory (black) for a spherical black hole. The singularity in the regular scenario at $\tau=0$ (signalled by $E_{3}=0$) is pushed indefinitely in negative $\tau$. The parameters are: $M=1$, $\Lambda=-0.1$, $\gamma=0.274$, $\theta=-0.3$, and $\beta=-0.1$. Full results may be found in \cite{ref:schneidnoncomm}.}}}
\label{fig:noncommsing}
\end{center}
\end{figure}

There are of course many scenarios that can be studied here, since not only can the geometry be spherical, toroidal, and higher-genus hyperbolic,  but also one can choose $\theta$ and $\beta$ to be zero or non-zero individually. We summarize the various results of \cite{ref:schneidnoncomm} in Table~1. 
\vspace{1.0cm}
%%\clearpage

\begin{center}
{\vspace{-0.00cm}\large \underline{\small{Table 1}}}\quad {\small{Summary of the non-commutative evolution of various black holes}}\enlargethispage{0.5cm}
\end{center}
{\small \begin{enumerate}\vspace{-0.5cm}
 \item {\bf Spherical}
\begin{itemize}
\item $\theta\neq 0,\,\beta=0$: Small values of $\theta$ delay the singularity. Large values of $\theta$ remove the singularity.
\item ${\theta=0,\,\beta\neq 0}$: Singularity is removed.
\item ${\theta\neq 0,\, \beta\neq 0}$: Singularity is removed.
\end{itemize}

\item {\bf Toroidal}
\begin{itemize}
\item $\theta\neq 0,\,\beta=0$: Small values of $\theta$ delay the singularity. Large values of $\theta$ introduce a new  horizon$^{*}$.
\item ${\theta=0,\,\beta\neq 0}$: Singularity is removed.
\item ${\theta\neq 0,\, \beta\neq 0}$: Singularity is removed for large $\theta$. For small $\theta$ a new horizon appears$^{*}$.
\end{itemize}
%%\vspace{0.80cm}
\item {\bf Higher-genus}
\begin{itemize}
\item $\theta\neq 0,\,\beta=0$: Small values of $\theta$ delay the singularity. Large values of $\theta$ introduce a new  horizon$^{*}$.
\item ${\theta=0,\,\beta\neq 0}$: Singularity is removed.
\item ${\theta\neq 0,\, \beta\neq 0}$: Singularity is removed for large $\theta$. For small $\theta$ a new horizon appears$^{*}$.
\end{itemize}
\end{enumerate}
\vspace{-0.1cm}
\hrule 
\vspace{0.1cm}{\footnotesize $^{*}$ The new horizon's presence prevents the determination of a singularity beyond the second horizon since the numerical scheme cannot penetrate horizons.}} 
\vspace{0.5cm}

A different non-commutative study in similar variables may be found in \cite{ref:othernoncomm}.

As in the previous section, we mention here some drawbacks to the scheme. First, from the results illustrated, the singularity is not necessarily removed in all scenarios. In particular, it seems non-commutativity in the canonical momenta is required in order to have singularity avoidance (i.e., $\beta\neq 0$). Non-commutativity in the configuration variables ($\theta\neq 0$) is insufficient. Also, some scenarios are inconclusive due to the formation of a second horizon. Also, adding the extra non-commutativity in the quantum mechanics of \emph{free} particles yields wave equations whose solutions mimic those of particles coupled to an electromagnetic field within regular quantum mechanics~\cite{ref:noncommqmech}. Therefore, non-commutativity seems to introduce some peculiar physical effects.

\subsection{Gravastars} \label{sec:gravastars}
One final alternative to singular spacetimes that we present in this brief review is an exotic object known as the gravastar (gravitational vacuum star). This model was first proposed by P. Mazur and E. Mottola \cite{ref:origgrava} and the idea is to not only eliminate the black hole singularity, but the horizon as well, so that there is no black hole. The model is essentially an extension of what were previously known as de~Sitter core black holes \cite{ref:gliner} - \cite{ref:poisisdesit}. 

The basic idea behind the gravastar is the following. Under extreme gravitational conditions, such as those found near the singularity of a black hole, the vacuum undergoes a phase transition and the spacetime becomes locally one of constant positive curvature. This is de~Sitter spacetime where the effective pressure equals minus the effective energy density. Note that this violates the inequality (\ref{eq:AsimpleSEC}) and therefore circumvents the singularity theorem of Appendix~\ref{sec:appendixA}. In the coordinate system of (\ref{eq:sphmet}) the de~Sitter spacetime possesses the line element
\begin{equation}
 ds^{2}=g_{\mu\nu}\textrm{d}x^{\mu}\textrm{d}x^{\nu}=-\left(1-\frac{\Lambda}{3}r^{2}\right)\textrm{d}t^{2} + \frac{\textrm{d}r^{2}}{1-\frac{\Lambda}{3}r^{2}}  +r^{2} \textrm{d}\theta^{2} + r^{2}\sin^{2}\theta\, \textrm{d}\phi^{2}\,, \label{eq:deSmet}
\end{equation}
with the cosmological constant $\Lambda > 0$. The Riemann tensor components in the orthonormal frame whose defining vectors $e^{\mu}_{\;\,\hat{\alpha}}$ point in the coordinate directions, are constant with magnitude $\Lambda/3$. Therefore, since in the vicinity of the black hole singularity ($r=0$) the spacetime is replaced by (\ref{eq:deSmet}), the curvature singularity does not exist in the gravastar model. If one moves the cosmological constant term in (\ref{eq:einsteq}) to the right-hand side of the equals sign, one can interpret the cosmological constant term in Einsteins equations as a stress-energy of the vacuum (since if $T_{\mu\nu}$, the stress-energy tensor of material, vanishes, there is still a source term on the right-hand side of the field equations). That is, in this interpretation the stress-energy of the vacuum is given by
\begin{equation}
 \underset{\mbox{\scriptsize{vac}}}{T_{\mu\nu}} = -\frac{\Lambda}{8\pi}g_{\mu\nu}\,. \label{eq:setvac}
\end{equation}
Therefore, in the de~Sitter scenario, which corresponds to $\Lambda >0$, the energy density, $-T^{0}_{\;\;\,0}$, is positive and the pressures, $T^{i}_{\;\;i}$ (no sum), are negative.

The occurrence of such a phase transition as described above is motivated by condensed matter physics where a stable low energy state (analogous to the gravitational vacuum here) phase transitions to another state at a critical temperature. Further discussion can be found in \cite{ref:mazmotPNAS} and an excellent review of gravastars in general may be found in \cite{ref:saibalgrava}.

As mentioned above, the gravastar proposal not only aims at the elimination of the singularity, but also the event horizon. However, there is now good observational evidence for black holes or at least objects that closely resemble black holes. It is therefore postulated that the phase transition occurs when gravitational conditions are such that a horizon is just about to form, and outside the gravastar the gravitational field is almost indistinguishable from the spacetime of the corresponding black hole that the gravastar replaces. Qualitatively the process would be something similar to the following. A star's nuclear fuel diminishes near the end of its life and, as described in section~\ref{subsec:formation}, gravitational contraction ensues. Assuming that nothing stops the contraction, as the material collapses in on itself the gravitational field at some points in spacetime approaches closer and closer to that of an event horizon. When the conditions for event horizon formation are almost satisfied, the gravitational phase transition takes over, so that the would-be black hole now evolves to a gravastar configuration. Therefore, just outside of where the horizon would have formed, the spacetime is almost exactly that of the corresponding black hole, making it difficult to distinguish a gravastar from a black hole using exterior observations. However it is not impossible to distinguish the two scenarios \cite{ref:gravadifferences}.

In the case of spherical symmetry, the interior de~Sitter spacetime (Equation (\ref{eq:deSmet})) needs to transition to an essentially Schwarzschild spacetime (Equation (\ref{eq:schwmet})) at   somewhere near but outside of the event horizon $r=2M$. In order for this to happen continuously, the radial pressure $T^{1}_{\;\;\,1}$ and the tangential pressures $T^{2}_{\;\;\,2}$ and $T^{3}_{\;\;\,3}$ must differ in some region along the transition zone~\cite{ref:catvis,ref:ourgrava}. This makes the gravastar anisotropic in parts of its interior region. This anisotropy is demanded by the Einstein equation (Equation (\ref{eq:einsteq})) governing the gravitational field. Furthermore, it is physically desirable that the outermost layer of the gravastar possess non-exotic physics. This means that the outer layer, which we call the atmosphere of the gravastar, should have a decreasing radial pressure and energy density, just like one expects in an ordinary star. The qualitative behavior of the continuous radial pressure which transitions from de~Sitter to Schwarzschild, as discussed here, is illustrated in Figure~\ref{fig:gravapres}. Continuity of the radial pressure in static spherical symmetry is demanded by the Synge junction condition~\cite{ref:syngebook}:
\begin{equation}
 \left[T^{\mu}_{\;\;\,\nu}\hat{n}^{\nu}\right]_{-} - \left[T^{\mu}_{\;\;\,\nu}\hat{n}^{\nu}\right]_{+} = 0\,, \label{eq:syngcond}
\end{equation}
when taking the unit normal $\hat{n}^{\nu}$ as pointing radially outward from an $r=\mbox{const.}$ junction surface. The $+$ and $-$ subscripts indicate whether one is approaching the junction surface from above or below. Note that since $T^{\mu}_{\;\;\,\nu}=0$ for the Schwarzschild metric, Equation (\ref{eq:syngcond}) implies that at the stellar surface, where the material terminates, the radial pressure must smoothly go to zero. This mathematically defines the boundary of the star.

\begin{figure}[htbp]
%%\begin{figure}[t!]
\begin{center}
\includegraphics[width=0.75\textwidth, height=0.60\textwidth, clip, keepaspectratio=false]{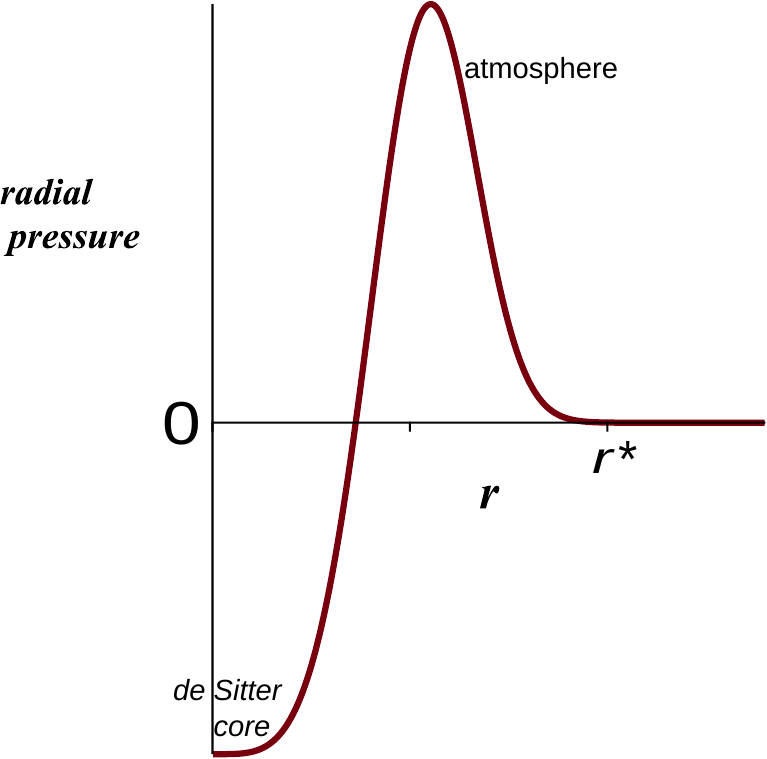}
\caption{{\small{The radial pressure, $T^{1}_{\;\;\,1}$, for a gravastar that is locally de~Sitter at $r=0$ and essentially Schwarzschild at $r>r^{*}$, taken as the boundary of the star. Since Schwarzschild is vacuum, this means the pressure must vanish there. As stated in the main text, it is demanded that in the atmosphere this pressure is monotonically decreasing. Full results may be found in \cite{ref:catvis}, \cite{ref:ourgrava}.}}}
\label{fig:gravapres}
\end{center}
\end{figure}

There has been a lot of work done on the gravastar alternative to black holes and a comprehensive review may be found in \cite{ref:saibalgrava}. Of particular importance in astrophysics are models with rotation. Perturbative rotating models have been studied in \cite{ref:rotgravstart}-\cite{ref:rotgravend} and a class of exact models has been constructed in \cite{ref:ourrotgrav} where an exactly de~Sitter interior has been smoothly patched to a Kerr (\ref{eq:kerrmet}) exterior with a physically acceptable atmosphere region. Rotating models are much more complicated to analyze so we omit a detailed discussion of them.

Again here we should close with some drawbacks to this model of singularity avoidance. One potential issue is that it is not clear why such a phase transition as required for gravastar formation should occur, nor exactly what the physics behind it would be. It is also not clear why the phase transition, even if it is feasible, should occur at the onset of horizon formation. It is true though that there are arguments that quantum phase transition effects are not only strong field effects, but can occur under relatively weak conditions \cite{ref:mazmotPNAS}, so the weakness of the gravitational field in the vicinity of a horizon does not necessarily preclude some sort of gravitationally induced quantum effects from occurring there.

%%%%%%%%%%%%%%%%%%%%%%%%%%%%%%%%%%%%%%%%%%
\section{Conclusions}

We have presented here a brief summary of the history of some singular spacetimes mainly within the theory of general relativity. It has long been desired to find an acceptable way to remove the singularities in these spacetimes and replace them with something more benign. This often requires the invocation of new or exotic physics and every alternative has its pros and cons. We have presented here only a few possibilities such as certain possible quantum gravity effects, non-commutative geometry effects, also inspired by quantum mechanics, and the gravastar alternative to black holes. All of these are admittedly speculative at the moment but the topic of gravitation is very difficult to test experimentally in the strong field regime so a certain degree of speculation is inevitable in the field.  As mentioned in the main text this is not an exhaustive review but will hopefully provide the reader with an understanding of the problem and possible avenues for a resolution.

\section*{Acknowledgments}The author would like to thank Professor S. Ray of the Centre for Cosmology, Astrophysics and Space Science (CCASS) at GLA University in Mathura, India for the invitation to participate in the conference associated with this manuscript. The author would also like to thank the anonymous referees for their suggestions, which improved the manuscript.

\PRLsep
\appendix
\setcounter{equation}{0}
\renewcommand\theequation{A.\arabic{equation}}

\section{Appendix: A simple singularity theorem} \label{sec:appendixA}
Here we present a simple singularity theorem which illustrates how certain types of singularities can occur. The understanding of this theorem is not critical to the rest of the manuscript and may be skipped without loss of continuity. A more in-depth analysis of such singularity theorems may be found in the classic book by Hawking and Ellis \cite{ref:hawkell}. We should note though that certain singularity theorems do not actually guarantee that a singularity will occur \cite{ref:newkerr}.

Before presenting the theorem and its proof, we state the strong energy condition. The strong energy condition asserts
\begin{equation}
 T_{\mu\nu}V^{\mu}V^{\nu} \geq \frac{1}{2} T\,V^{\mu}V_{\mu} \quad \forall \;\mbox{timelike}\;\; V^{\mu}\,, \label{eq:ASEC}
\end{equation}
where $T_{\mu\nu}$ is the stress-energy tensor of the gravitating material and $T:=T^{\alpha}_{\;\;\,\alpha}$ its invariant trace. If the stress-energy tensor is diagonalizable, and the ``zeroth'' coordinate is timelike and the other three spacelike (Hawking-Ellis type one \cite{ref:hawkell}) then the stress-energy tensor can be represented as
\begin{equation}
 \left[T^{\mu}_{\;\;\,\nu}\right]=\mbox{diag}\left[ -\rho, \, p_{1},\, p_{2},\, p_{3}\right]\,, \label{eq:Adiagset}
\end{equation}
where $\rho$ represents the energy density of the material in its rest frame, and the $p_{i}$ its principal pressures, that is, momentum flux. For this structure of stress-energy tensor the strong energy condition (\ref{eq:ASEC}) simplifies to
\begin{equation}
 \rho+p_{i}\geq 0 \qquad \mbox{and}\qquad \rho +\sum_{i=1}^{3} p_{i} \geq 0\,. \label{eq:AsimpleSEC}
\end{equation}

Now, if the Einstein field equations (\ref{eq:einsteq}) with $\Lambda=0$ hold, it can easily be shown by tracing the Einstein equations that $T=-R/8\pi$ and with this result that 
\begin{equation}
 R_{\mu\nu} =8\pi \left(T_{\mu\nu}-\frac{1}{2}T\,g_{\mu\nu}\right)\,. \label{eq:ARofT}
\end{equation}
In terms of the Ricci tensor therefore, the strong energy condition (\ref{eq:ASEC}) can be recast as
\begin{equation}
 R_{\mu\nu}V^{\mu}V^{\nu}\geq 0\,. \label{eq:ARicecond}
\end{equation}
This statement of non-negative Ricci curvature physically implies that in general relativity matter should not create a repulsive gravitational effect for timelike test particles.

\vspace{0.1cm}The singularity theorem is now stated as \cite{ref:dasdebbook}:
\begin{theorem}
Assume that the following field equations hold in a domain $D\subset \mathbb{R}^{4}$:
\begin{equation}
 G_{\mu\nu}=8\pi T_{\mu\nu}\, \label{eq:Aeinsteq}
\end{equation}
and further, that $T_{\mu\nu}$ is diagonalizable as in (\ref{eq:Adiagset}). Let the strong energy condition (\ref{eq:ASEC}) hold. Assume also that the vorticity tensor vanishes and further suppose that a congruence of timelike velocities, $V^{\mu}$, possesses a negative expansion scalar $\theta(x) < 0$ in a neighborhood of $D$. Then, each of the timelike geodesics $x=\chi(s)$ will terminate at a negative expansion singularity at $x^{*}=\chi(s^{*})\subset \partial D$  after a finite proper time.
\end{theorem}
\begin{proof}
Consider the Raychaudhuri equation with vanishing vorticity,
\begin{equation}
\frac{\textrm{d}\theta(s)}{\textrm{d}s} +\frac{1}{3}\theta^{2}(s)= -\left\{ \sigma^{\mu\nu}\sigma_{\mu\nu} + R_{\mu\nu}V^{\mu}V^{\nu}\right\}_{|x=\chi(s)}\,, \label{eq:ARay}
\end{equation}
where $\sigma_{\mu\nu}$ is the shear tensor whose square, $\sigma^{\mu\nu}\sigma_{\mu\nu}$, can be shown to be non-negative \cite{ref:MBlau}. From the field equations (\ref{eq:ARofT}) holds so, via (\ref{eq:Adiagset}), the Raychaudhuri equation can be written as
\begin{equation}
 \frac{\textrm{d}\theta(s)}{\textrm{d}s} +\frac{1}{3}\theta^{2}(s)= -\left\{ \sigma^{\mu\nu}\sigma_{\mu\nu} + 4\pi \left(\rho +\sum_{i} p_{i} \right) \right\}_{|x=\chi(s)}\;. \label{eq:ARayT}
\end{equation}
If the strong energy condition holds the right-hand side of the above is less than or equal to zero so that
\begin{equation}
  -\left\{ \sigma^{\mu\nu}\sigma_{\mu\nu} + 4\pi \left(\rho +\sum_{i} p_{i} \right) \right\}_{|x=\chi(s)} =: -\left[f(s)\right]^{2} \leq 0\,, \label{eq:Aconverge}
\end{equation}
and multiplying (\ref{eq:ARayT}) by $\theta^{-2}(s)$ allows us to write it as:
\begin{equation}
 \frac{\textrm{d}}{\textrm{d}s}\theta^{-1}(s)=\frac{1}{3}+[f(s)]^{2}[\theta(s)]^{-2} \geq 0\,. \label{eq:Aconverge2}
\end{equation}

Now, assume that there is an initial convergence $\theta(s_{0}) < 0$. The solution to the differential equation (\ref{eq:Aconverge2}) is
\begin{equation}
\theta^{-1}(s)= \theta^{-1}(s_{0}) + \frac{1}{3}(s-s_{0}) + \int_{s_{0}}^{s} [f(s')]^{2}[ \theta(s')]^{-2}\, \textrm{d}s'\,. \label{eq:Aconvergesol}
\end{equation}
Assuming continuity and differentiability, from the monotonicity of (\ref{eq:Aconvergesol}) there will be some finite {$s=s^{*} > s_{0}$} where  $\theta^{-1}(s^{*})=0$\, or\, $\theta(s^{*}) \rightarrow -\infty$ and there will be infinitely strong contraction.
\end{proof}

\PRLsep

\vspace{-0.080cm}

%%\vspace{0.5cm}
%%\newpage
\linespread{0.6}
\bibliographystyle{unsrt}

%%\PRLsep
\end{document}